\documentclass{svproc}

\usepackage{url}

\usepackage{fancyhdr}
\usepackage{amssymb}
\usepackage{mathrsfs} 
\usepackage{amsmath}
\usepackage{color}

\def\QED{\protect{\flushright{$\Box$}}}
\def\root{\mathop{\rm root}}
\def\Ker{\mathop{\rm Ker}}
\def\rank{\mathop{\rm rank}}
\def\bbR{\mathbb{R}}
\def\bbN{\mathbb{N}}

\usepackage{graphicx}
\def\figfile{./nct}
\def\mpfig#1{\DeclareGraphicsRule{.#1}{mps}{*}{} 
			 \vcenter{\hbox{\includegraphics{\figfile.#1}}}}

\begin{document}
\mainmatter              

\title{Trees in the Real Field}

\author{Alessandro Betti\inst{1}\inst{2} \and Marco Gori\inst{2}}

\tocauthor{Alessandro Betti, Marco Gori}

\institute{University of Florence, Florence, Italy,\\
\email{alessandro.betti@unifi.it}
\and
SAILab,
University of Siena, Siena, Italy,\\
\email{marco@diism.unisi.it}\\
WWW home page: \texttt{http://sailab.diism.unisi.it}
}

\maketitle              

\begin{abstract}
This paper proposes an algebraic view of trees 
which opens the doors to an alternative computational scheme 
with respect to classic algorithms. In particular, it is shown that
this view is very well-suited for machine learning and computational
linguistics.
\end{abstract}

\section{Introduction}
In the last few years models of deep learning have been successfully applied
to computational linguistics. Amongst others, the translation problem has
benefited
significantly from simple approaches based on recurrent neural networks. In particular,
because of the classic problem of capturing long-term dependencies~\cite{Bengio_trnn93}, 
LSTM~\cite{Hochreiter:1997:LSM} architectures have been mostly used which can better deal with this classic problem. 

In this paper we go beyond this approach and assume to characterize linguistic 
production by means of generative trees by relying on the principle that the complexity of the 
problem of long-term dependencies is dramatically reduced because of the 
exponential growth of nodes of the trees with respect to their height. 
In general the relations between trees and their corresponding linear encoding is 
not easy to grasp. For example, when restricting to binary trees, it can be proven
that we need a pair of traversals to fully characterize a given tree, one of which may be
the symmetric one~\cite{taocp1}. However, whenever a sequence presents a
certain degree of regularity, the ambition arises to establish a bijection with a
corresponding tree (e.g. the parsing tree).

While encoding mechanisms are quite straightforward to design every time that
it is possible to assign to each sequence a tree-like structure; it is
sufficient to propagate the information (for example with a linear
scheme) through the nodes up to the root of the
tree~(\cite{frasconi1998general}), it is much harder
to came up with a decoding scheme that generates the translated
sequence.
Here we prove that we can construct
a decoding scheme that naturally extend those used nowadays in
recurrent neural nets that can be potentially very interesting in computational
 linguistics.

 \section{Uniform real-valued tree representations}
A binary tree is recursively defined as
\begin{equation}
  {\tt T}=
  \begin{cases}
    {\tt T}_{\emptyset} &        \text{basis}\\
     ({\tt L},{\tt y},{\tt R}) & \text{induction} 
    \end{cases}
\label{TreeDef}
\end{equation}
where ${\tt T}_{\emptyset}$ is the empty tree,
${\tt y} \in \Sigma$ is the labeled root, 
which takes on values from the alphabet $\Sigma$,
{\tt L} (Left), and {\tt R} (Right) are trees.
We assume that we are given a {\em coding
function} $\ell: \Sigma \rightarrow \mathscr{Y}  \subset \bbR^{p}$, so as
the nodes of the tree are related to an associated point\footnote{In the 
following we will often regard the elements of $\tt T$ as elements 
of $\bbR^p$, without mentioning function $\ell$ explicitly.}
 of $\mathscr{Y}$.
Now, let us consider the the pair
\begin{eqnarray}
\label{RealValuedTR}
 	&i.& T:=\left(L,x,R\right)\\
\nonumber
	&ii.& \gamma: \mathscr{X} \subset \bbR^{n} \rightarrow \mathscr{Y},
              \quad
	\gamma(x):=Cx
\end{eqnarray}
which consists of the triple $\left(L,x,R\right)$ and of the 
{\em linear labeling function} $\gamma$, which returns points, that will be
related to the labels of {\tt T}.
In the triple, we have 
$x \in \bbR^{n}, \ L,R \in \bbR^{n\times n}$. Basically, we
introduce a computational scheme on the {\em embedding space} $\mathscr{X}$.
If $x = 0$ then we assume that
$\left(L,0,R\right) \sim (0,0,0) := T_{\emptyset}$.
We want to explore the relations between the tree definition~(\ref{TreeDef})
and the related real-valued representation given by equations~(\ref{RealValuedTR}). 
To this end, we start noticing that
the {\tt void} tree ${\tt T}_{\emptyset}$ can be associated with $T_{\emptyset}$.
The idea is that we can specify a tree ${\tt T}$ once the triple $(L,x,R)$ and $C$ 
are given. Beginning from $Cx=\root(\tt T)$, we process the children of
the root by applying $L$ and $R$ to $x$, 
so that $CRx$ is the right child and $CLx$ is the left
child. Then the left child of the left child of the root is
obtained as $CLLx$, and the right child of the left child of the root
as $CRLx$, and so on and so forth, until we find, for each branch of the tree, a node
$l \in \bbR^{n}$ for which $Ll=Rl=0$. This will be the leaf of that particular
path, and we will say that the children of the leaves are buds; more
generally every null node will be denoted as a bud.

We say that $(L,x,R)$ is an $n$-dimensional {\it real                           
representation\/} of the obtained tree $\tt T$.

In order to get an insight on this construction let us consider the following examples.

\begin{example}
The first non-trivial example is the         
tree that consists  of the root only. In our representation this tree is
obtained by picking up any two matrices $L$ and $R$, such
that $x$ in their kernel, that is $L x = R x = 0$.
The simplest next example is given by
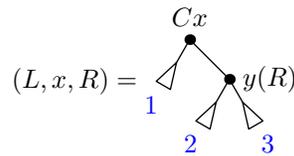
\[(L,x,R)=\mpfig1\]
The decoding equations that defines this tree are     
\[\begin{cases}Cx=\root(\tt T);\\ CRx=y(R),\end{cases}\quad
  \begin{cases}CLx=0,& \text{bud {\color{blue}$1$}};\\          
                        CLRx=0,& \text{bud {\color{blue}$2$}};\\     
                        CR^2x=0,& \text{bud {\color{blue}$3$}},\end{cases}\]
They are conveniently separated into the  ``node conditions''  and ``bud conditions''.
In order to be even more explicit consider the case $C=I$, $x=(1,0)'$ 
and $y(R)=(0,1)'$, then it is easy to check that
\begin{equation}                                
  \mpfig2=\left(\begin{pmatrix}0&0\\ 0&0\end{pmatrix},\begin{pmatrix} 1\\ 0
    \end{pmatrix},                      
\begin{pmatrix}0&0\\1&0\end{pmatrix}\right). \label{ex1}
\end{equation}                               
We can easily see that in this special case, this representation 
is unique in $\bbR^{2}$.
\label{SimplestEx}                                                     
\end{example}

As soon as we think about the next example with two nodes $\mpfig3$,
a symmetry property of the decoding scheme becomes evident. 
Given $\tt T$ let us define the symmetric left-right
${\tt T}'$ as the tree that one obtains from $\tt T$ by 
recursively exchanging the left with the right subtrees. For
example
$${\tt T}=\mpfig4,\qquad {\tt T}'=\mpfig5,$$
are related by the defined symmetry operation.
Clearly, for those trees we can state an immediate property
on their representation.

\begin{proposition}
Let $\tt T$ and ${\tt T}'$ be related by {\em left-right symmetry} and let
$(L,x,R)$ be a real representation of $\tt T$. Then $(R,x,L)$ is the
representation of ${\tt T}'$.
\end{proposition}

\begin{proof}
Straightforward.\QED
\end{proof}

This result immediately shows us when looking at the tree given by~(\ref{ex1}),
that we have
$$\mpfig6=\left(\begin{pmatrix}0&0\\ 1&0\end{pmatrix},\begin{pmatrix}1\\0
  \end{pmatrix},                      
\begin{pmatrix}0&0\\ 0&0\end{pmatrix}\right).$$

 \begin{example} 
 	In this case  we show the role of the embedding space 
	$\mathscr{X} \subset \bbR^{n}$. In particular, we will see that the decoding
	might not be solvable at certain dimensions and that there could be also
	infinite solutions.
	Let us consider the following tree with the associated {\it decoding equations}
	\[(L,x,R)=\mpfig8\quad,\qquad
	\begin{cases}Cx=\root(\tt T);\\ CLx=y(L);\\ CRLx=y(RL),\end{cases}
        \quad 
	\begin{cases}CL^2x=0,& \text{bud {\color{blue}$1$}};\\  
		   CLRLx=0,& \text{bud {\color{blue}$2$}};\\ 
		   CR^2Lx=0, & \text{bud {\color{blue}$3$}};\\       
           CRx=0,& \text{bud {\color{blue}$4$}}.\end{cases}\]  
    We consider two different cases $n=2$ and $n=3$.   

	\begin{itemize}
	\item {\bf Case $n=2$. \enspace}
	Let us consider $n=2$ and assume $C=I$. In addition, let us assume that
	the nodes of $\tt T$ are coded by 
	\begin{eqnarray*}
		\root({\tt T}) = (1,0)^{\prime}, \ y(L) = (0,1)^{\prime}, 
		\ y(RL) = (1,1)^{\prime}.
	\end{eqnarray*}
	From $C L^{2} x = 0$ and from $C L x = y(L)$
	we get $L (Lx) = 0$, that is $L y(L)=0$. This yields 
	a constraint on the structure of $L$; we have
	\[
		\begin{pmatrix}
			l_{11} & l_{12}\\
			l_{21} & l_{22}\end{pmatrix}
		\cdot\begin{pmatrix}0\\ 1\end{pmatrix}
		= \begin{pmatrix} 0\\ 0\end{pmatrix}
		\to
		L =\begin{pmatrix}
			l_{11} & 0\\
			l_{21} & 0\end{pmatrix}. 	
	\]
	Likewise from $C R L x = y(RL)$ we get
	\[
	\begin{pmatrix} r_{11} & r_{12}\\
			r_{21} & r_{22}\end{pmatrix}
		\cdot\begin{pmatrix} 0\\ 1\end{pmatrix}
		= \begin{pmatrix} 1\\ 1\end{pmatrix}
		\to
		R = \begin{pmatrix}
			r_{11} & 1\\
			r_{21} & 1\end{pmatrix} 
	\]
	From $C L R L x = 0$ we get
	\begin{eqnarray}
	\label{partialres1}
		l_{11}(r_{11} l_{11} + l_{21})= 0\\
	\label{partialres1}
		l_{21}(r_{21} l_{11}  + l_{21})= 0
	\end{eqnarray}
	Now, let $x = (x_{1},x_{2})^{\prime}$ be. From $L x=y(L)$  we get 
	$l_{11} x_{1} = 0$ and $l_{21} x_{1} = 1$. Then $l_{11} = 0$, which, in turn, 
	satisfies~(\ref{partialres1}). Then, from~(\ref{partialres1}) we get
	$l_{21}=0$. Then, we end up into an impossible satisfaction of
	 $l_{21} x_{1} = 1$. 
	 
	 \item {\bf Case $n=3$. \enspace}
	 Let us consider $n=3$ and still assume $C=I$. In addition, let us assume that
	the nodes of $\tt T$ are coded by 
\[ \root({\tt T}) = (1,0,0)^{\prime}, \ y(L) = (0,1,0)^{\prime}, 
		\ y(RL) = (0,0,1)^{\prime}.\]
	From $C x = \root({\tt T})$ we get $x = (1,0,0)$.
	From $C L^{2} x = 0$ and from $C L x = y(L)$
	we get $L (Lx) = 0$, that is $L y(L)=0$. 
	This yields 
	a constraint on the structure of $L$; we have
	\[
		\begin{pmatrix}
			l_{11} & l_{12} & l_{13}\\
			l_{21} & l_{22} & l_{23}\\
			l_{31} & l_{32} & l_{33}\end{pmatrix}
                      \cdot\begin{pmatrix}0\\ 1\\ 0\end{pmatrix}
		=\begin{pmatrix}0\\ 0\\ 0\end{pmatrix}
		\to
		L = 
	\begin{pmatrix}
			l_{11} & 0  & l_{13}\\
			l_{21} & 0 & l_{23}\\
			l_{31} & 0 & l_{33}\end{pmatrix}. 	
	\]
	From $L x = y(L)$ we get
	\[      \begin{pmatrix}
			l_{11} & 0  & l_{13}\\
			l_{21} & 0 & l_{23} \\
			l_{31} & 0 & l_{33}
		\end{pmatrix}
		\cdot 
		\begin{pmatrix}
		1 \\
		0 \\
		0
		\end{pmatrix}
		=
		\begin{pmatrix}
		0 \\
		1 \\
		0
		\end{pmatrix}
	\]
	that is $l_{21}=1$ and $l_{11}=l_{31}=0$.	
	Likewise from $C R L x = y(RL)$ we get
	\[
		\begin{pmatrix}
			r_{11} & r_{12} & r_{13} \\
			r_{21} & r_{22} & r_{23} \\
			r_{31} & r_{32} & r_{33}
		\end{pmatrix}\cdot
		\begin{pmatrix}
		0 \\ 1 \\ 0
		\end{pmatrix}
		= 
		\begin{pmatrix}
		0 \\ 0 \\ 1
		\end{pmatrix}\to
		R = 
		\begin{pmatrix}
			r_{11} & 0 & r_{13} \\
			r_{21} & 0 & r_{23} \\
			r_{31} & 1 & r_{33}
		\end{pmatrix}.
	\]
	From $C L R L x = 0$ we get
	\[
		\begin{pmatrix}
			0 & 0  & l_{13}\\
			1 & 0 & l_{23} \\
			0 & 0 & l_{33}
		\end{pmatrix}
		\cdot
		\begin{pmatrix}
			r_{11} & 0 & r_{13} \\
			r_{21} & 0 & r_{23} \\
			r_{31} & 1 & r_{33}
		\end{pmatrix}
		\cdot
		\begin{pmatrix}
			0 & 0  & l_{13}\\
			1 & 0 & l_{23} \\
			0 & 0 & l_{33}
		\end{pmatrix}
		\cdot
		\begin{pmatrix}
			1 \\
			0 \\
			0
		\end{pmatrix}
		=
		\begin{pmatrix}
			0 \\
			0 \\
			0
		\end{pmatrix}
	\]
	Hence, 
	\[
		\begin{pmatrix}
			0 & 0  & l_{13}\\
			1 & 0 & l_{23} \\
			0 & 0 & l_{33}
		\end{pmatrix}
		\cdot
		\begin{pmatrix}
			0\\
			0 \\
			1
		\end{pmatrix}
		=
		\begin{pmatrix}
			0 \\
			0 \\
			0
		\end{pmatrix},
	\]
	which is  satisfied if $l_{13}=l_{23}=l_{33}=0$.
	
	From $C R^{2} L x = 0$ we get
	\[
		\begin{pmatrix}
			r_{11} & 0 & r_{13} \\
			r_{21} & 0 & r_{23} \\
			r_{31} & 1 & r_{33}
		\end{pmatrix}
		\cdot 
		\begin{pmatrix}
			r_{11} & 0 & r_{13} \\
			r_{21} & 0 & r_{23} \\
			r_{31} & 1 & r_{33}
		\end{pmatrix}
		\cdot 
		\begin{pmatrix}
			0 \\
			1 \\
			0
		\end{pmatrix}
		=
		\begin{pmatrix}
			0 \\
			0 \\
			0
		\end{pmatrix}
		\to
		\begin{pmatrix}
			r_{11} & 0 & r_{13} \\
			r_{21} & 0 & r_{23} \\
			r_{31} & 1 & r_{33}
		\end{pmatrix}
		\cdot 
		\begin{pmatrix}
			0 \\
			0 \\
			1
		\end{pmatrix}
		=
		\begin{pmatrix}
			0 \\
			0 \\
			0
		\end{pmatrix}.
	\]
	Finally, from $R x = 0$ we need $r_{11}=0$. Then
        we conclude that $L$ and $R$ are solutions 
	whenever they have the structure
	\[
		L =
		\begin{pmatrix}
			0 & 0  & 0\\
			1 & 0 & 0 \\
			0 & 0 & 0
		\end{pmatrix}\qquad
		R = 
		\begin{pmatrix}
			0 & 0 & 0 \\
			r_{21} & 0 & 0 \\
			r_{31} & 1 & 0
		\end{pmatrix}.
	\]
	Notice that in this case we discover infinite solutions. 
	In addition, it is worth mentioning that this solution originates from the required
	labeling, since it immediately requires to choose $x=(1,0,0)$. This 
	makes it possible to satisfy the matrix monomial equations without
	requiring strong nilpotent conditions on the matrices. In addition, in
	this case, there is no solution for any $x$, since otherwise we 
	need to require $R=0$. As a consequence, the other labelling conditions
	would not be met.
	If we assume to keep a representation based on the above matrices $L,R$
	then a different choice of $x$ may led to a completely different tree.
	For example, we can easily see that the choices $x=(0,1,0)^{\prime}, 
	(0,0,1)^{\prime}$ yield infinite trees. 
\end{itemize}
\label{ThreeNodeEx}
\end{example}
Interestingly, the generation of infinite trees is not an exception, but quite a 
common property of the introduced generative scheme. 

Let us consider a simple example that clearly shows the
possible explosion of the introduced generation scheme. 
Let us consider a tree whose
elements are two dimensional vectors, and consider a two dimensional
representation; in addition, for the sake of simplicity, let us assume that
 $C=I$ and $\root({\tt T})=(1,0)'$. Then let us assume that $R$ is a $\pi$ rotation and $L$
is a projection onto the $y$ axis:\DeclareGraphicsRule{.7}{mps}{*}{}
{\parshape 6 0pt 23pc 0pt 24pc 0pt 25pc 0pt 26pc 0pt 27pc 0pt 24pc\parfillskip=0pt
\smash{\raise-5pt\rlap{\kern5pt
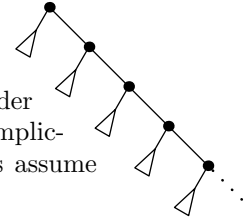
{\includegraphics{\figfile.7}}}}\par}
\[x=\begin{pmatrix}1\\0\end{pmatrix},\quad
  L=\begin{pmatrix}0&0\\0&1\end{pmatrix},\quad
R=\begin{pmatrix}-1&0\\0&-1\end{pmatrix}.\]
An infinite tree with flipping labels is generated that is
shown in  the side figure.

As shown in the previous examples,
we are interested in solving equations involving  monomials of matrices.
 Let us focus on the algebraic side and consider the following example.
\begin{example}
	Let us consider the monomial equation
	\begin{eqnarray}
		L R = 0.
	\label{2NipPotLR}
	\end{eqnarray}
	What are the non-null matrices $L$ and $R$ which satisfy this equation?
	Clearly, equations like $L^{2}=0$ and $R^{2}=0$ define nilpotent matrices
	of order $2$. Equation~(\ref{2NipPotLR}) can be regarded as a sort
	of generalization of the notion of nilpotent matrix to the case in which
	the property involves two matrices. 
	
	This problem has generally infinite solutions. Any pair of matrices $L$, $R$
	such that the image space of $R$ is in the kernel of $L$ is a solution.
	The pair $L=\bigl({-2\atop 2}{1\atop -1}\bigr)$ and $R=\bigl({1\atop 2}{-2\atop -4}\bigr)$
	is an example. The image space of $R$ is in the kernel of $L$. Of course,
	matrix $R$ must be singular, otherwise its image space would invade 
	the whole $\bbR^{2}$ and $\Ker(A) = \left\{ 0 \right\}$, which would require
	matrix $L=0$.
\end{example}
As discussed in Example~\ref{ThreeNodeEx}, in general we 
need the  satisfaction of monomial equations that also involve $x \in \bbR^{n}$.

\begin{example} 
Suppose we are given $T=(x,L,R)$ where $L=\bigl({2\atop 2}{-2\atop -2}\bigr)$ and $R=\bigl({1\atop 1}{-1\atop -1}\bigr)$.
We can promptly see that $L^{2}=R^{2}=0$, and $[L, R] = L R - R L = 0.$
The last one comes out in any case in which $R = \alpha L$, with $\alpha \in \bbR$ (here $\alpha = 1/2$).
We can immediately conclude that any pair $(L,R)$, where $L^{2}=0$ and $R = \alpha L$
corresponds with a balanced tree composed of three nodes. 
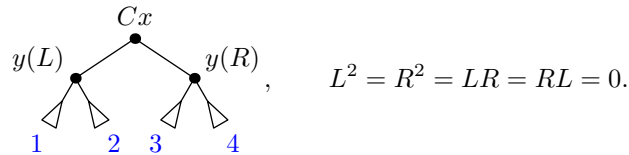
\[\mpfig9,\qquad L^{2}=R^{2}=L R = RL=0.\]
Notice that in order to define the formal correspondence with this non-void balanced tree
we need to restrict to the condition $x \not \in \Ker L$. 
On the opposite, if we choose $x = \beta (1,1)'$
with $\beta \in \bbR \setminus\left\{ 0 \right\}$ then the triple represents a tree composed of the root only. 
If $x = 0$ then the triple degenerates to one of the infinite representations of the {\tt void} tree.

Now, let us consider the problem of mapping the above tree in the 
representation $(L,x,R)$. We need to match the labels $\root({\tt T}), \ y(L)$
and $y(R)$. Hence we must impose:
\[C x = \root({\tt T}),\qquad
	C L x = y(L),\qquad
	C R x =  y(R).\]
Since $R = \alpha L$ we have $y(R) = \alpha C L x = \alpha y(L)$.
This clearly indicates that while the representation 
$(L,x,\alpha L)$ is a balanced tree, there is a strong restriction on the
label that it can produce. 
\label{InsightCompleteT}
\end{example}

\subsubsection{Paths and monomial correspondence.}
The discussion on the representation of trees in the real field given in the
previous examples enlightens on a nice connection between paths
and monomials. In order to decode a certain node we generally need
to associate nodes with monomials like
\[ L, \quad R, \quad L^{2}, \quad LR, \quad RL, \quad R^{2},   \quad L^{3}, \quad L^{2} R,\quad R L^{2},\quad R^{3},\quad LRL,\quad RLR, \ldots
\]
composed with the two variables $L$ and $R$. This kind of monomials turn out
to be just another way of expressing a path in a tree.
The above monomial are of
degree $3$, but we are interested in monomials of any order, which can 
be represented by the language generated with symbols $L$ and $R$.
For instance, the sequence
\[
LRLLRLLLLRRLRLRLR = (LR) \cdot 
(L^{2}) \cdot (R^{1})  \cdot
(L^{2})^{2}\cdot(R^{2})
\cdot(LR)^{3}
\]
is a way of constructing a monomial with $L$ and $R$, that could also be
regarded as an element of the language generated by $S_{1}=R$, $S_{2}=L^2$,
$S_{3}=LR$. 
This monomials can be described as follows. Let 
$\ell_\nu$ and $r_\nu$ be the integer vectors that count the repetitions of $L$
and $R$ is the sequence, respectively. In the above sequence we have
\begin{eqnarray*}
	\ell^\nu &= (1,2,4,1,1,1) \\
	r^\nu &= (1,1,2,1,1,1).
\end{eqnarray*}
This notation makes is possible to express the sequence as
\[
	\pi^{\nu}= LRLLRLLLLRRLRLRLR
	:= L^{(1,2,4,1,1,1)} R^{(1,1,2,1,1,1)} = L^{\ell^\nu} R^{r^\nu},
\]
where we assume that the above path characterizes node $\nu$.
Consistently with what we have done so far will indicate the label on the node
$\nu$ with the notation $y(\pi^\nu)\in\mathscr{Y}$.
Here, the notations $L^{\ell^\nu} R^{r^\nu}$
reminds us of a generalized notion of
matrix power for the matrices $L$ and $R$. The notation used for $\pi^{\nu}$
reminds the characterization of the node $\nu$, while the generic arc of the
path $\pi^{\nu}$ is simply an element $\pi^{\nu}_\kappa$ of vector $\pi^{\nu}$.
Moreover, we also use the notation 
$|\ell^{\nu}| = \sum_{\kappa} \ell^\nu_\kappa$ and 
$|r^{\nu}| = \sum_{\kappa} r^\nu_\kappa$.
Clearly $|\pi^{\nu}|=|\ell^{\nu}| + |r^{\nu}|$.

Example~\ref{InsightCompleteT} gives an insight to draw the following general conclusion
\begin{proposition}
	Let $\alpha \in \bbR$ and $R = \alpha L$ be. Moreover, let us
assume that $h \in \bbN$ and $h \geq 1$ is the first integer such
$L^{h} = 0$. If $x \not \in \Ker L^{h-1}$ then the decoding of the triple $T=(L,x,R)$ is a balanced tree {\tt T} with height $h$.
\label{BalParalCond}
\end{proposition}

\begin{proof}
	The proof can be given straightforwardly by induction on $h$.
\QED
\end{proof}

The possible generation of infinite trees raises the question on which conditions
we need to impose in order to gain the guarantee that a given representation
yields finiteness. In addition to the condition stated in Proposition~\ref{BalParalCond},
in the next section we will present another class of representations which gives
rise to finite tree. The following proposition states a general property on the
generation of ``vanishing trees". 

\begin{proposition}
	Given $T=(L,x,R)$ let us assume that $\Vert L \Vert <1, \ 
	\Vert R \Vert <1$. 
	Then if $\nu$ is a leaf of path $\pi^\nu$
	\[
	\lim_{|\pi| \rightarrow \infty} L^{\ell^{\nu}} R^{r^{\nu}} x = 0.
	\] 
\label{VanishingProp}
\end{proposition}
\begin{proof}
	We have 
	\[
		y(\pi^\nu) = C L^{\ell^{\nu}} R^{r^{\nu}} x. 
	\]
	When taking the norm on both sides
\[
		\Vert y(\pi^\nu) \Vert \leq 
		\Vert C \Vert \cdot
		\Vert L^{\ell^\nu} R^{r^\nu} \Vert
		\cdot 
		\Vert x \Vert
\]
	Now, let $\theta <1$ be an upper bound of $\Vert L \Vert$ 
	and $\Vert R \Vert$.
	Then
\[
		\Vert y_{\nu} \Vert \leq 
		\Vert C \Vert \cdot
		\Vert x \Vert
		\cdot \theta^{|\pi|}.
\]
	Finally, the proof follows when computing $\lim_{|\pi| \rightarrow \infty}$.
\QED
\end{proof}

We are now ready to formulate the decoding problem in its general form.

\subsubsection{Decoding Problem.} 
Given the tree 
${\tt T}$ with $m$ nodes we consider the equations
\[
  \begin{cases}
    CL^{\ell^\nu}R^{r^\nu}x=y(\pi^\nu)&\text{for all nodes $\nu$};\\
    CL^{\ell^\beta}R^{r^\beta}x=0&\text{for all buds $\beta$},
    \end{cases}
\]
which refers to the nodes and to the buds, respectively
(remember that a binary tree with $m$ nodes has $m+1$ buds). 
When using the vectorial form, we ca rewrite this conditions
in the form $Mx=y$ where
\[
M:=  \begin{pmatrix}
    CL^{\ell^1}R^{r^1}\\
    CL^{\ell^2}R^{r^2}\\
    \vdots\\
    CL^{\ell^m}R^{r^m}\\
    CL^{\ell^{m+1}}R^{r^{m+1}}\\
    \vdots\\
     CL^{\ell^{2m+1}}R^{r^{2m+1}}
   \end{pmatrix}
  \qquad\hbox{and}\qquad
y:=  \begin{pmatrix}
    y(\pi^1)\\
    y(\pi^2)\\
    \vdots\\
    y(\pi^m)\\
    0\\
    \vdots\\
    0
   \end{pmatrix}.
\]
\begin{definition}
The representation $(L,x,R)$ of\/  ${\tt T}$
is completely reachable if and only if\/ $\rank M= \min \left\{n, p \cdot
  (2m+1)\right\}$.
\end{definition}

\begin{proposition}
	Let us consider any completely reachable pair $(L,R)$ 
	of ${\tt T}$. If  $n \geq p \cdot (2m+1)$ then
	the decoding problem
	of\/ ${\tt T}$ admits the solution 
	\[
		x = M^+ y,
	\]
	where $M^+$ is Penrose pseudo-inverse of $M$.
\end{proposition}

\section{Non-commutative left-right matrices}
%
%
As we have already seen, $T$ can yield an infinite tree. Here is another example.
\begin{example}
Let us consider the triple $T=\left(L,x,R \right)$ where $L=\bigl({2\atop 4}{-1\atop -2}\bigr)$ and $R=\bigl({1\atop 1}{-1\atop -1}\bigr)$.
We can promptly see that $L^{2}=0$ and $R^{2}=0$, but  $[L, R] \neq 0$.
In particular
\[
		[L,R] = 
		\begin{pmatrix}
			2 & -1 \\
			4 & -2
		\end{pmatrix}
		\cdot
		\begin{pmatrix}
			1 & -1 \\
			1 & -1
		\end{pmatrix}
		- 
		\begin{pmatrix}
			1 & -1 \\
			1 & -1
		\end{pmatrix}
		\cdot
		\begin{pmatrix}
			2 & -1 \\
			4 & -2
		\end{pmatrix}
		= 
		\begin{pmatrix}
			3 & -2 \\
			4 & -3
		\end{pmatrix}
\]
We can easily check that the recursive propagation yields an infinite tree.
\end{example}

No matter whether a finite or an infinite tree is generate, a uniform representation
$(T,\gamma)$ is especially interesting whenever $[L,R] \neq 0$.
In the opposite case, as already seen, the representation is dramatically limited. 
The following example suggests to consider a nice class of uniform 
non-commutative representations.
%
%
The following example shows a representation $(L,x,R)$ which yields
finite trees.

\begin{example}
	Let us consider the triple $T=\left(L,x,R \right)$ where
\[
		L = \begin{pmatrix}
			0 & 0 & 0 \\
			b_{l} & 0 & 0 \\
			a_{l} & c_{l} & 0\end{pmatrix}
\qquad
		R = \begin{pmatrix}
			0 & 0 & 0 \\
			b_{r} & 0 & 0 \\
			a_{r} & c_{r} & 0
		\end{pmatrix}
\]
Let $a_{l},b_{l},c_{l},a_{r},b_{r},c_{r}$ be non-null reals and associate any non-null
real with symbol $\odot$.
Then we have
\[
\begin{split}
	|\pi^\nu|=2 \to  \pi^\nu &= 
               \begin{pmatrix}
			0 & 0 & 0 \\
			\odot & 0 & 0 \\
			\odot & \odot & 0
		\end{pmatrix}\cdot
		\begin{pmatrix}
			0 & 0 & 0 \\
			\odot & 0 & 0 \\
			\odot & \odot & 0
		\end{pmatrix}=
		\begin{pmatrix}
			0 & 0 & 0 \\
			0 & 0 & 0 \\
			\odot  & 0 & 0
		\end{pmatrix}
	\\
	|\pi^\nu|=3 \to \pi^\nu &= 
		\begin{pmatrix}
			0 & 0 & 0 \\
			0 & 0 & 0 \\
			\odot & 0 & 0
		\end{pmatrix}\cdot 
		\begin{pmatrix}
			0 & 0 & 0 \\
			\odot & 0 & 0 \\
			\odot & \odot & 0
		\end{pmatrix}= 0
\end{split}\]
This corresponds with the balanced tree in Fig.~1.

\begin{figure}[t]
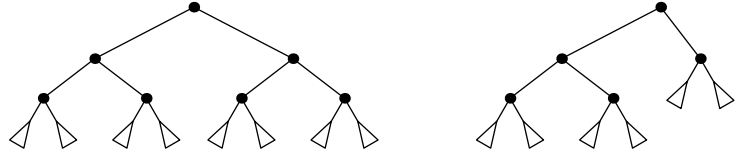

\DeclareGraphicsRule{.10}{mps}{*}{}\DeclareGraphicsRule{.11}{mps}{*}{}
\vbox{\hfil\includegraphics{\figfile.10}\hfil
\includegraphics{\figfile.11}\hfil}
\caption{Balanced tree on the left in the case of non-null coefficients.
  If $b_{r}=0$ then the asymmetry yields the unbalanced tree on the right.}
\end{figure}

Now, we can exploit the non-commutativity $[L,R] \neq 0$ to generate
other trees with missing nodes. We easily see that if $b_{r}=0$ then
$RL=R^{2}=0$ (see Fig.~1).
\end{example}

\section{Conclusions}
The encoding-decoding scheme presented in this paper opens the doors
to new learning algorithms that seem to be adequate in computational linguistics.
A different path that may be followed is the one of restricting to commuting matrices
where different matrices are used for any layer. 

\section*{Acknowledgments}
We thank Ilaria Cardinali for insightful discussions. 

\bibliography{nn}
\bibliographystyle{plain}

\end{document}